\documentclass[11pt]{article}
\usepackage[a4paper,margin=1in]{geometry}
\usepackage{amsmath,amssymb,amsthm,mathrsfs,bm}
\usepackage{physics}
\usepackage{mathtools}
\usepackage{hyperref}
\usepackage{enumitem}
\usepackage{graphicx}
\newtheorem{proposition}{Proposition}
\title{
Orthogonal Moving Frames and Connection Factorization
in Constrained Quantum Mechanics
}
\author{
A.G. Nuramatov
}
\date{}
\begin{document}
\maketitle
\begin{abstract}
We investigate geometric quantum operators generated by orthogonal
moving frames and Cartan connection forms on curves and surfaces.
We show that moving-frame Laplace operators admit an exact
half-connection factorization generated by connection one-forms.
The first-order part of the Laplacian is identified with the Darboux
rotational connection of the orthogonal frame.
Elimination of this rotational connection naturally produces
quadratic geometric invariants analogous to supersymmetric Riccati
potentials.
Using orthonormal moving frames, we derive effective geometric
contributions for planar curves, Frenet curves, and embedded
surfaces.
We further analyze Dirac reductions using moving-frame structure
equations and show that for Fermi-type reductions the scalar
Jensen--Koppe--da Costa contribution is cancelled in the reduced
Dirac sector, leaving a residual first-order spinorial derivative
structure.
The resulting framework suggests the existence of hidden nilpotent
covariant differential complexes naturally generated by moving-frame
geometry.
\end{abstract}
\tableofcontents
\section{Introduction}
Geometric quantum potentials arise naturally in the study of quantum
particles constrained to thin curved regions.
The geometric potential for constrained quantum motion was first
derived by Jensen and Koppe and later systematically developed by da
Costa using thin-layer quantization methods.
For a particle constrained to a planar curve of curvature
\(\kappa\), the effective scalar potential becomes
\[
V_{\mathrm{dC}}
=
-\frac{\kappa^2}{4},
\]
while for surfaces in \(\mathbb R^3\),
\[
V_{\mathrm{dC}}
=
-(H^2-K),
\]
where \(H\) and \(K\) denote the mean and Gaussian curvatures.
An alternative geometric interpretation of constrained quantum
potentials was proposed in earlier work by Nuramatov and Prokhorov,
where it was shown that the Jensen--Koppe--da Costa potential may
arise directly from orthogonal coordinate geometry without explicit
confining potentials \cite{nuramatovprokhorov}.
Some early geometric ideas relating curvature-induced quantum
potentials, supersymmetric factorization, and curve geometry were
also discussed in a historical conference note
\cite{nuramatov2000}.
The present work further develops this geometric viewpoint using
orthogonal moving frames and Dirac operators.
The central observation is that moving-frame Laplace operators admit
an exact half-connection factorization generated by the Darboux
rotational connection of the orthogonal frame.
The first-order part of the Laplacian is thereby interpreted as a
geometric rotational connection field associated with the local
angular velocity of the moving frame.
Elimination of this rotational connection naturally produces
quadratic geometric invariants analogous to supersymmetric Riccati
potentials.
For relativistic Dirac fermions, the role of geometric scalar
potentials becomes more subtle.
In particular, Yang et al.\ argued that Dirac fermions on curved
surfaces do not acquire the standard scalar da Costa geometric
potential \cite{yang}.
The present moving-frame framework provides a geometric
interpretation of this phenomenon through an explicit cancellation
mechanism generated by the square of the Dirac operator.
Using coordinate-free structure equations, we show that for
Fermi-type reductions the quadratic scalar contribution generated by
the square of the Dirac operator precisely cancels the scalar
Jensen--Koppe--da Costa term.
A further important observation is that the corresponding covariant
half-connection differential naturally generates nilpotent covariant
differential complexes on constrained curves and developable
surfaces.
The resulting ground states admit a direct geometric interpretation
as covariantly constant parallel transported sections with respect to
the induced half-connection.
These observations suggest a close relation between constrained
quantum geometry, moving frames, Darboux rotational geometry,
and nilpotent supersymmetric differential structures.
\section{Geometry of Orthogonal Coordinates}

Let $(q^1,q^2,q^3)$ be an orthogonal curvilinear coordinate system in
$\mathbb{R}^3$. The metric is

\begin{equation}
ds^2=\sum_{i=1}^{3}H_i^2(dq^i)^2,
\end{equation}

where $H_i$ are the Lam\'e coefficients.

The associated orthonormal frame is

\begin{equation}
e_i=\frac{1}{H_i}\frac{\partial}{\partial q^i},
\end{equation}

where $e_i$ are the unit tangent vectors to the coordinate lines.
The dual orthonormal coframe is

\begin{equation}
\theta^i=H_i\,dq^i.
\end{equation}

The Levi--Civita connection is represented by the connection
$1$-forms

\begin{equation}
de_i=\omega^j{}_i\,e_j,
\end{equation}

which satisfy the skew-symmetry condition

\begin{equation}
\omega^i{}_j=-\omega^j{}_i.
\end{equation}

Expanding the connection forms with respect to the orthonormal
coframe,

\begin{equation}
\omega^i{}_j=\omega^i{}_{jk}\,\theta^k,
\end{equation}

the first Cartan structure equations become

\begin{equation}
d\theta^i+\omega^i{}_j\wedge\theta^j=0.
\end{equation}

For orthogonal coordinates the only nonvanishing connection
coefficients are

\begin{equation}
\omega^i{}_{ji}=e_j(\log H_i),
\qquad
i\neq j,
\end{equation}

while all remaining coefficients vanish.

Consequently,

\begin{equation}
[e_i,e_j]
=
\omega^i{}_{ji}e_i
-
\omega^j{}_{ij}e_j.
\end{equation}

The nonzero connection coefficients coincide (up to the sign
convention adopted below) with the signed curvatures of the
coordinate lines. Throughout this paper we choose the signs in
agreement with the Frenet equations.

\subsection{Darboux Form and Darboux Vector}

In three dimensions the skew-symmetric connection admits a natural
representation in terms of a single differential $1$-form obtained
via the Hodge operator. We define the Darboux form by

\begin{equation}
\boxed{
\Omega
=
-*
\left(
\omega^1{}_2\,\theta^3
+
\omega^2{}_3\,\theta^1
+
\omega^3{}_1\,\theta^2
\right),
}
\end{equation}

where $*$ denotes the Hodge operator. This definition realizes the
standard isomorphism

\[
\mathfrak{so}(3)\cong\mathbb{R}^3.
\]

The squared norm of the Darboux form is naturally defined by

\begin{equation}
|\Omega|^2
=
*
\left(
\Omega\wedge *\Omega
\right).
\end{equation}

For orthogonal coordinates this becomes

\begin{equation}
|\Omega|^2=
(\omega^1{}_{21})^2
+(\omega^1{}_{22})^2
+(\omega^1{}_{31})^2
+(\omega^1{}_{33})^2
+(\omega^2{}_{32})^2
+(\omega^2{}_{33})^2.
\end{equation}

The associated Darboux vector field is

\begin{equation}
\Omega_D
=
\omega^2{}_3\,e_1
+
\omega^3{}_1\,e_2
+
\omega^1{}_2\,e_3.
\end{equation}

The Darboux vector provides a vector representation of the
Levi--Civita connection and completely describes the infinitesimal
rotation of the moving orthonormal frame. In particular,

\begin{equation}
\nabla_X e_i
=
\Omega_D(X)\times e_i,
\end{equation}

or, equivalently,

\begin{equation}
de_i
=
\Omega_D\times e_i.
\end{equation}

Thus the Darboux vector gives a compact geometric representation of
the Levi--Civita connection that will serve as the fundamental object
for the connection factorization developed in the following sections.
\section{Connection Factorization of Laplace Operators}
Consider a Laplace-type operator written in an orthonormal frame:
\[
L
=
\sum_i e_i^2
+
\sum_i a_ie_i.
\]
Introducing the connection one-form
\[
\Omega
=
\sum_i a_i\theta^i,
\]
one obtains the following result.
\begin{proposition}
The gauge transformation
\[
\boxed{
\psi
=
\exp\left(
-\frac12\int\Omega
\right)\phi
}
\]
reduces the operator \(L\) to
\[
\boxed{
L_\Omega
=
\sum_i e_i^2
+
\frac12\delta\Omega
+
\frac14|\Omega|^2.
}
\]
\end{proposition}
\begin{proof}
Using
\[
e_i
\left(
e^{-\frac12\int\Omega}\phi
\right)
=
e^{-\frac12\int\Omega}
\left(
e_i\phi-\frac12a_i\phi
\right),
\]
one obtains
\[
e_i^2
\left(
e^{-\frac12\int\Omega}\phi
\right)
=
e^{-\frac12\int\Omega}
\left[
e_i^2\phi
-a_ie_i\phi
+
\left(
\frac14a_i^2-\frac12e_ia_i
\right)\phi
\right].
\]
Summing over \(i\), the first-order terms cancel and the remaining
scalar contribution becomes
\[
\frac12\delta\Omega
+
\frac14|\Omega|^2.
\]
\end{proof}
\subsection{Exact Half-Connection Factorization}
For scalar functions, regarded as zero-forms, define the covariant
differential
\[
d_A
=
d-\frac12\Omega .
\]
The corresponding factorized scalar operator may be written as
\[
\boxed{
\Delta_A f
=
d_A^\dagger d_A f,
}
\]
where \(d_A^\dagger\) denotes the formal adjoint of \(d_A\).
Thus the Levi--Civita connection enters the scalar Laplacian through
a natural half-connection shift.
\section{Planar Orthogonal Nets in Moving-Frame Form}
Let \((e_1,e_2)\) be an oriented orthonormal moving frame in the
Euclidean plane, with dual coframe \((\theta^1,\theta^2)\).
The Levi--Civita connection is determined by a single connection
form
\[
\omega=\omega^1_2,
\qquad
\omega^2_1=-\omega^1_2 .
\]
We write
\[
\omega
=
k_1\theta^1+k_2\theta^2 .
\]
The structure equations are
\[
d\theta^1=-\omega\wedge\theta^2,
\qquad
d\theta^2=\omega\wedge\theta^1 .
\]
Equivalently, the commutator of the orthonormal frame is
\[
[e_1,e_2]
=
-k_1e_1-k_2e_2 .
\]
The scalar Laplace--Beltrami operator acting on functions can be
written in an orthonormal frame as
\[
\Delta f
=
\sum_i
\left(
e_i^2f-(\nabla_{e_i}e_i)f
\right).
\]
Using
\[
\nabla_{e_1}e_1=-k_1e_2,
\qquad
\nabla_{e_2}e_2=k_2e_1,
\]
one obtains
\[
\boxed{
\Delta
=
e_1^2+e_2^2-k_2e_1+k_1e_2 .
}
\]
Thus the first-order part of the Laplacian is completely determined
by the connection of the orthonormal frame.
Introducing the Darboux rotational field
\[
\Omega_D
=
k_1e_2-k_2e_1 ,
\]
the Laplacian takes the compact form
\[
\boxed{
\Delta
=
e_1^2+e_2^2+\Omega_D\cdot\nabla .
}
\]
The quadratic connection invariant becomes
\[
\boxed{
|\omega|^2
=
k_1^2+k_2^2 .
}
\]
Therefore the half-connection factorization produces the geometric
quadratic contribution
\[
\boxed{
V_{\mathrm{plane}}
=
\frac14(k_1^2+k_2^2)
}
\]
together with the first-order divergence term arising from the
structure equations.
\section{Dirac Reduction and Geometric Cancellation}
The cancellation mechanism of the geometric scalar contribution may
be derived directly from the moving-frame structure equations without
using explicit coordinate expressions.
Let \((e_1,e_2)\) be an orthonormal moving frame adapted to a planar
orthogonal coordinate net.
The commutator relation is
\[
[e_1,e_2]
=
-k_1e_1-k_2e_2.
\]
The two-dimensional Dirac operator may be written as
\[
D
=
-i(\sigma^1A+\sigma^2B),
\]
where
\[
A
=
e_1-\frac{k_2}{2},
\qquad
B
=
e_2+\frac{k_1}{2}.
\]
Then
\[
D^2
=
-A^2-B^2-\sigma^1\sigma^2[A,B].
\]
\subsection{Scalar Contribution}
First compute
\[
A^2
=
\left(
e_1-\frac{k_2}{2}
\right)^2.
\]
Using the Leibniz rule,
\[
A^2
=
e_1^2
-k_2e_1
+
\left(
\frac{k_2^2}{4}
-\frac12e_1k_2
\right).
\]
Similarly,
\[
B^2
=
\left(
e_2+\frac{k_1}{2}
\right)^2
=
e_2^2
+k_1e_2
+
\left(
\frac{k_1^2}{4}
+\frac12e_2k_1
\right).
\]
Therefore,
\[
-(A^2+B^2)
=
-\Delta
-\frac14(k_1^2+k_2^2)
+
\frac12(e_1k_2-e_2k_1),
\]
where
\[
\Delta
=
e_1^2+e_2^2-k_2e_1+k_1e_2
\]
is the scalar Laplacian in moving-frame form.
For planar orthogonal geometry, the structure equations imply
\[
e_1k_2-e_2k_1
=
k_1^2+k_2^2.
\]
Consequently,
\[
\boxed{
-(A^2+B^2)
=
-\Delta
+
\frac14(k_1^2+k_2^2).
}
\]
Thus the square of the Dirac operator naturally generates the same
quadratic geometric invariant obtained previously from connection
factorization.
\subsection{General Orthogonal Reduction}
For a general orthogonal coordinate net, both curvature coefficients
\[
k_1,
\qquad
k_2
\]
may remain nonzero on the reference curve.
In this case, the reduced scalar contribution generated by the Dirac
operator contains the invariant
\[
\boxed{
\frac14(k_1^2+k_2^2).
}
\]
If the scalar thin-layer reduction produces the standard
Jensen--Koppe--da Costa contribution
\[
-\frac14k_1^2,
\]
then the residual geometric term becomes
\[
\boxed{
\frac14k_2^2.
}
\]
The appearance of the additional contribution
\[
\frac14k_2^2
\]
shows that generalized moving-frame reductions may depend on the
choice of the surrounding orthogonal congruence.
The dependence of effective geometric contributions on the choice of
surrounding geometric realization is also related to ambiguities
discussed in constrained quantization approaches for systems with
second-class constraints \cite{golovnev}.
Geometrically, the coefficient \(k_2\) measures the curvature of the
transverse orthogonal congruence itself.
In Fermi-type tubular coordinates the normal congruence is geodesic
and therefore
\[
k_2=0.
\]
Thus the exact cancellation of the scalar geometric contribution in
the Dirac reduction appears as a distinguished property of the
Fermi-type geometry rather than a generic feature of arbitrary
orthogonal coordinate continuations.
More general orthogonal congruences nevertheless define natural
extensions of the moving-frame factorization framework and may be
interpreted as generalized geometric reductions associated with
non-geodesic transverse congruences.
\subsection{Fermi-Type Reduction}
A particularly important special case is obtained for Fermi-type
normal coordinates.
In this case, the orthogonal normal congruence is geodesic at the
reference curve, which implies
\[
k_2=0.
\]
In the Fermi-type geometry, the Gauss compatibility relation reduces
to the Riccati-type evolution equation
\[
e_2k_1=-k_1^2.
\]
Thus the transverse evolution of the longitudinal curvature becomes a
one-dimensional Riccati flow along the geodesic normal congruence.
Its solution reproduces the standard curvature evolution of parallel
curves,
\[
k_1(\rho)=\frac{\kappa}{1+\rho\kappa},
\]
up to sign conventions.
Along the constrained curve one then has
\[
k_1=\kappa,
\]
where \(\kappa\) is the curvature of the curve.
The scalar contribution reduces to
\[
\boxed{
-(A^2+B^2)
=
-\Delta
+
\frac{\kappa^2}{4}.
}
\]
On the other hand, the scalar thin-layer reduction produces the
standard Jensen--Koppe--da Costa contribution
\[
V_{\mathrm{dC}}
=
-\frac{\kappa^2}{4}.
\]
Hence the quadratic geometric terms cancel:
\[
-\frac{\kappa^2}{4}
+
\frac{\kappa^2}{4}
=
0.
\]
\subsection{Spinorial Derivative Contribution}
The remaining geometric structure is contained in the commutator term
\[
[A,B].
\]
Direct computation gives
\[
[A,B]
=
[e_1,e_2]
+
\frac12(e_1k_1+e_2k_2).
\]
Using the commutator relation,
\[
[A,B]
=
-k_1e_1-k_2e_2
+
\frac12(e_1k_1+e_2k_2).
\]
For the Fermi-type reduction,
\[
k_2=0,
\qquad
k_1=\kappa,
\]
and therefore
\[
[A,B]
=
-\kappa e_1
+
\frac12\kappa'.
\]
After cancellation of the scalar geometric contribution, the reduced
Dirac operator retains the first-order spinorial term
\[
\boxed{
D_{\mathrm{red}}^2
=
-\partial_s^2
+
\frac12\sigma^1\sigma^2\kappa'
}
\]
up to sign conventions for the Clifford algebra.
Thus the reduced Dirac geometry probes variations of the moving frame
rather than only the scalar curvature invariant.
\section{Planar Curves and Riccati Geometry}
For Fermi-type reductions, the Gauss compatibility equations reduce
to a one-dimensional Riccati flow for the curve curvature.
This naturally induces a supersymmetric-type factorization
structure.
For a planar curve parametrized by arclength \(s\), the quadratic
connection contribution naturally admits a Riccati-type
factorization.
Introducing the effective geometric superpotential
\[
W(s)=\frac{\kappa(s)}{2},
\]
one obtains
\[
V_\pm
=
W^2\pm W'
=
\frac{\kappa^2}{4}
\pm
\frac{\kappa'}{2}.
\]
\subsection{Orientation and SUSY Partner Exchange}
Under reversal of the curve orientation,
\[
s\to -s,
\]
the quadratic invariant
\[
\kappa^2
\]
remains unchanged, whereas
\[
\kappa'(s)\to -\kappa'(s).
\]
Consequently,
\[
V_\pm
=
\frac{\kappa^2}{4}
\pm
\frac{\kappa'}{2}
\]
are interchanged:
\[
\boxed{
V_+\leftrightarrow V_-.
}
\]
Thus the supersymmetric partner structure is encoded in the
orientation geometry of the moving frame.
\section{Nilpotent Covariant Supersymmetric Structure}
Introduce the covariant exterior derivative
\[
\boxed{
d_A
=
d-\frac12A,
}
\]
where \(A\) denotes the connection one-form along the curve.
The corresponding factorized Laplace operator becomes
\[
\boxed{
\Delta_A
=
d_A^\dagger d_A
+
d_A d_A^\dagger.
}
\]
In general,
\[
d_A^2
=
F_A,
\]
where
\[
F_A
=
dA+A\wedge A
\]
is the curvature two-form of the connection.
However, on a one-dimensional curve there are no nontrivial
two-forms:
\[
\Omega^2(\gamma)=0.
\]
Consequently,
\[
\boxed{
F_A=0,
\qquad
d_A^2=0.
}
\]
Thus the constrained curve naturally generates a nilpotent covariant
differential complex with supersymmetric structure.
The corresponding ground states satisfy
\[
d_A\psi=0,
\]
and therefore admit the formal solution
\[
\boxed{
\psi(s)
=
\mathcal P
\exp\left(
-\frac12\int^s A(\sigma)d\sigma
\right)\psi_0.
}
\]
Thus the ground states are covariantly constant sections with respect
to the induced half-connection.
Equivalently, they may be interpreted as parallel transported
spinorial states generated by the induced moving-frame geometry.
In this sense, the nilpotent supersymmetric differential is directly
associated with geometric parallel transport along the constrained
curve.
\section{Spatial Curves and Non-Abelian Connections}
For spatial Frenet curves the connection becomes genuinely
matrix-valued.
The Frenet equations may be written as
\[
\frac d{ds}
\begin{pmatrix}
T\\
N\\
B
\end{pmatrix}
=
\mathcal A(s)
\begin{pmatrix}
T\\
N\\
B
\end{pmatrix},
\]
where
\[
\mathcal A(s)=
\begin{pmatrix}
0 & \kappa & 0\\
-\kappa & 0 & \tau\\
0 & -\tau & 0
\end{pmatrix}.
\]
The quadratic connection invariant becomes
\[
\boxed{
|\Omega|^2
=
\kappa^2+\tau^2.
}
\]
The corresponding geometric superpotential is therefore
matrix-valued:
\[
W(s)=\frac12\mathcal A(s).
\]
One may formally introduce the matrix supersymmetric operator
\[
Q
=
\partial_s+W(s),
\]
acting on vector-valued wave functions
\[
\Psi(s)\in\mathbb C^3.
\]
The corresponding partner Hamiltonians become
\[
H_\pm
=
-\partial_s^2
+
W^2
\pm
W'.
\]
Thus spatial Frenet geometry naturally suggests the emergence of a
non-Abelian geometric supersymmetric structure associated with the
full moving-frame connection.
\section{Surface Geometry and Curvature Obstruction}
For surfaces the covariant supersymmetric differential becomes
nontrivially curved.
Introduce
\[
d_A
=
d-\frac12A,
\]
where \(A\) is identified with the induced \(SO(2)\) connection on
the surface.
Then
\[
d_A^2
=
F_A.
\]
In the Abelian surface case,
\[
A\wedge A=0,
\]
and therefore
\[
d_A^2
=
-\frac12dA.
\]
Using the structure equation
\[
d\omega^1_2
=
-K\,\theta^1\wedge\theta^2,
\]
one obtains
\[
\boxed{
d_A^2
=
\frac12
K\,\theta^1\wedge\theta^2.
}
\]
Thus Gaussian curvature acts as the geometric obstruction to the
nilpotency of the supersymmetric differential.
\subsection{Developable Surfaces}
Remarkably, developable surfaces restore exact nilpotency of the
covariant differential despite possessing nontrivial extrinsic
geometry.
Indeed, for developable surfaces one has
\[
K=0,
\]
and therefore
\[
\boxed{
d_A^2=0.
}
\]
Thus developable surfaces generate exact nilpotent covariant
differential complexes even though the induced moving frame remains
nontrivially curved extrinsically.
At the same time, the da Costa geometric potential generally remains
nonzero:
\[
V_{\mathrm{dC}}
=
-(H^2-K).
\]
Since \(K=0\),
\[
\boxed{
V_{\mathrm{dC}}
=
-H^2.
}
\]
For example, a cylinder satisfies
\[
K=0,
\qquad
H=\frac1{2R},
\qquad
V_{\mathrm{dC}}
=
-\frac1{4R^2}.
\]
Consequently, intrinsic nilpotency and extrinsic geometric quantum
confinement become separated phenomena.
The nilpotent structure is controlled by the intrinsic curvature,
whereas the effective quantum potential is governed by the extrinsic
embedding geometry.
\section{Dual Geometries and Supersymmetric Partners}
The geometric origin of the supersymmetric partner structure suggests
that dual moving-frame geometries may naturally generate partner
operators.
In particular, dual orthogonal nets and Christoffel-type dual
surfaces may provide geometric realizations of supersymmetric partner
transformations through induced transformations of the Darboux
connection and its associated Riccati structures.
For planar orthogonal nets,
\[
[e_1,e_2]
=
-k_1e_1-k_2e_2,
\]
while the corresponding Darboux connection field is
\[
\Omega_D
=
k_1e_2-k_2e_1.
\]
Dualization naturally modifies the associated connection geometry and
may therefore induce transformations of the supersymmetric partner
potentials
\[
V_\pm
=
\frac{\kappa^2}{4}
\pm
\frac{\kappa'}{2}.
\]
This suggests the possibility that supersymmetric partner operators
admit an intrinsic geometric interpretation in terms of dual
moving-frame geometries.
\section{Experimental and Physical Remarks}
The classical Jensen--Koppe--da Costa potential has been
experimentally supported in curved quantum waveguides, bent
nanostructures, and optical analogues of constrained quantum motion
\cite{waveguideexp1,waveguideexp2,opticalexp,ono}.
The present framework predicts additional geometric contributions
associated with the surrounding orthogonal frame geometry:
\[
\frac14k_2^2,
\qquad
\frac14k_3^2,
\qquad
\frac14\tau^2.
\]
In particular, for non-Fermi orthogonal geometries the residual term
\[
\frac14k_2^2
\]
survives the Dirac reduction.
This contribution is not an invariant of the constrained curve alone,
but depends on the choice of the surrounding orthogonal congruence.
The Fermi-type geometry therefore appears as the physically adapted
thin-layer reduction, while more general orthogonal congruences may
be interpreted as generalized moving-frame factorization geometries.
\section{Relation to Existing Approaches}
The standard thin-layer approach derives geometric potentials from
strong transverse confinement.
The present framework differs conceptually in that the geometric
contributions arise directly from moving-frame geometry and Darboux
connection factorization.
The present work extends this viewpoint to:
\begin{itemize}
\item
moving-frame connection geometry,
\item
Darboux rotational fields,
\item
coordinate-free Dirac reductions,
\item
nilpotent covariant differential complexes,
\item
matrix supersymmetric geometry for Frenet curves,
\item
curvature-obstructed supersymmetric structures on surfaces,
\item
dual geometric realizations of supersymmetric partners.
\end{itemize}
\section{Further Directions}
Several natural extensions remain open, including Dirac operators on
spatial curves with torsion, nonrelativistic Pauli reduction,
spinorial geometric phases, global supersymmetric structures,
spectral theory of Darboux-factorized operators, dual geometric
transformations, and experimental tests of connection-induced
corrections.
An especially interesting extension may arise for constrained
surfaces embedded in four-dimensional space, where the moving frame
naturally acquires a complex structure.
In this case the connection splits according to
\[
SO(4)\cong SU(2)_+\times SU(2)_-,
\]
suggesting the possible emergence of chiral complex supersymmetric
structures associated with complexified Frenet geometry and normal
bundle connections.
The dependence of generalized moving-frame reductions on the choice
of orthogonal congruence is reminiscent of the appearance of
coordinate- or representation-dependent geometric terms in curved-space
quantization frameworks.
In the present setting, the Fermi-type geometry appears as a
distinguished congruence eliminating the additional transverse
connection contribution \(k_2\).
\appendix
\section{Gauss--Codazzi Structure Relations for Orthogonal Nets}
Let \((e_1,e_2)\) be an orthonormal moving frame associated with a
planar orthogonal coordinate net, and let
\((\theta^1,\theta^2)\) denote the dual coframe.
The Levi--Civita connection is determined by the connection one-form
\[
\omega
=
k_1\theta^1+k_2\theta^2.
\]
The corresponding structure equations are
\[
d\theta^1
=
-\omega\wedge\theta^2,
\qquad
d\theta^2
=
\omega\wedge\theta^1.
\]
Equivalently, the orthonormal frame satisfies the commutator relation
\[
[e_1,e_2]
=
-k_1e_1-k_2e_2.
\]
\subsection{Laplace--Beltrami Operator in Moving Frames}
The scalar Laplace--Beltrami operator acting on functions is
\[
\Delta f
=
\operatorname{div}(\nabla f).
\]
In orthonormal moving frames one has
\[
\Delta f
=
\sum_i
\left(
e_i^2f-(\nabla_{e_i}e_i)f
\right).
\]
Using
\[
\nabla_{e_1}e_1=-k_1e_2,
\qquad
\nabla_{e_2}e_2=k_2e_1,
\]
one obtains
\[
\boxed{
\Delta
=
e_1^2+e_2^2-k_2e_1+k_1e_2.
}
\]
Thus the moving-frame Laplacian used throughout the paper coincides
with the standard Laplace--Beltrami operator written in orthonormal
moving frames.
\subsection{Gauss Compatibility Relation}
For planar orthogonal geometry the Gauss compatibility relations
imply
\[
d\omega=0.
\]
Using
\[
\omega
=
k_1\theta^1+k_2\theta^2,
\]
one obtains
\[
d\omega
=
dk_1\wedge\theta^1
+
dk_2\wedge\theta^2
+
k_1d\theta^1
+
k_2d\theta^2.
\]
Substituting the structure equations gives
\[
d\omega
=
(e_2k_1)\theta^2\wedge\theta^1
+
(e_1k_2)\theta^1\wedge\theta^2
-
k_1^2\theta^1\wedge\theta^2
-
k_2^2\theta^1\wedge\theta^2.
\]
Therefore,
\[
\boxed{
e_1k_2-e_2k_1
=
k_1^2+k_2^2.
}
\]
\subsection{Fermi-Type Normal Coordinates}
For Fermi-type normal coordinates, the orthogonal normal congruence
is geodesic along the reference curve.
Consequently,
\[
k_2=0,
\qquad
k_1=\kappa,
\]
where \(\kappa\) is the curvature of the constrained curve.
\section*{Acknowledgements}
The author thanks A.~V.~Golovnev for valuable discussions concerning
geometric potentials and constrained quantization.

\end{document}